\documentclass[11pt]{amsart}

\usepackage{amssymb, amsmath, amsthm, mathtools, floatflt, palatino, euler, epic, eepic, microtype}
\usepackage{comment}
\usepackage{cases}
\usepackage{graphicx}
\usepackage[colorlinks, linkcolor=darkred, citecolor=darkblue, pagebackref=true, pdftex]{hyperref}
\usepackage[linesnumbered,ruled]{algorithm2e} 
\usepackage[normalem]{ulem}
\usepackage[usenames]{xcolor}
\usepackage[margin=1 in]{geometry}
\usepackage{subcaption}

\makeatletter
\def\BState{\State\hskip-\ALG@thistlm}
\makeatother

\definecolor{darkblue}{rgb}{0,0,0.7}
\definecolor{darkred}{rgb}{0.7,0,0}

\newcommand{\R}{\mathbb{R}}

\newtheorem{theorem}{Theorem}
\newtheorem{Lemma}[theorem]{Lemma}

\newtheorem{Corollary}[theorem]{Corollary}
\newtheorem{Definition}[theorem]{Definition}

\newtheorem{example}[theorem]{Example}
\newtheorem{Remark}[theorem]{Remark}

\newtheorem*{Theorem*}{Theorem}
\newtheorem*{cor*}{Corollary}
\newtheorem*{prop*}{Proposition}
\newtheorem*{conjecture*}{Conjecture}

\begin{document}

\title[Intersection properties of finite disk collections]{Intersection properties of finite disk collections}


\author[J. F. Espinoza]{Jes\'us F. Espinoza}
\email{jesusfrancisco.espinoza@unison.mx}

\author[Cynthia G. Esquer]{Cynthia G. Esquer-P\'erez}
\email{cynthia.esquer@unison.mx}
\address{Departamento de Matem\'aticas, Universidad de Sonora, M\'exico}

\begin{abstract}
{In this work we study the intersection properties of a finite disk system in the euclidean space. We accomplish this by utilizing subsets of spheres with varying dimensions and analyze specific points within them, referred to as poles. Additionally, we introduce two applications: estimating  the common scale factor for the radii that makes the re-scaled disks intersects in a single point, this is the \v{C}ech scale, and constructing the minimal Axis-Aligned Bounding Box (AABB) that encloses the intersection of all disks in the system.}
\end{abstract}

\maketitle


\section{Introduction}
One of the new techniques developed for the analysis of large clusters of information, known as Big Data, is Topological Data Analysis (TDA). In TDA, simplicial complexes associated with the data are constructed. These structures include the Vietoris-Rips complex, the \v{C}ech complex, and the piecewise linear lower star complex, among others. Of special interest to us is the generalized \v{C}ech complex structure. Although the standard \v{C}ech complex is formed by intersecting a collection of disks with a fixed radius, the generalized version allows varying radii. This flexibility enables us to highlight specific data points by assigning or weighting them with larger and/or more rapidly expanding balls to them, while de-emphasizing others by using smaller and/or slower growing balls. This approach proves valuable for handling noisy data sets, offering an alternative to discarding data that may not meet a specific significance threshold \cite{Bell_2019}.

Understanding the patterns of intersections and the timing of intersections among a set of disks in $\mathbb{R}^d$, each with potentially different radii, is a fundamental problem. This leads to the exploration of the generalized \v{C}ech complex structure, which captures the intersection information of these disks, regardless of their radii. Rescaling the radii by the same factor, we obtain a filtered generalized \v{C}ech complex, where the associated simplicial complexes evolve as the scale parameter varies. In particular, in \cite{le2015construction}, algorithms are provided to calculate the generalized \v{C}ech complex in $\mathbb{R}^2$, and \cite{Espinoza2019} presents an algorithm to determine the \v{C}ech scale for a collection of disks in the plane.

To establish the necessary foundation for our study, Section  \ref{Section:Disk System Intersections} introduces crucial concepts and notation that will be used throughout the article and we focus on analyzing the intersection of a disk system in $\mathbb{R}^d$. We start by investigating the intersection of two disks in Subsection \ref{Subsection:two.sphere.system} and then expand our analysis to a system of $m$ disks in Subsection \ref{Subsection:more.sphere.system}. By applying Helly's Theorem, we prove that it is sufficient to examine the intersection of all subsystems consisting of $d+1$ disks in order to determine if the system has a nonempty intersection.  

 In Section \ref{Section:Rips.Cech.systems}, we define  Vietoris-Rips systems and \v{C}ech systems, together with presenting results regarding the Rips scale and the \v{C}ech scale, as well as their connections. In Subsection \ref{Subsection:Cech.system}, we present an algorithm that can determine whether the intersection of the system is empty or non-empty. This is achieved by exclusively computing the poles  of subsystems of disks (or spheres). Finally, in Subsection \ref{Section:cech.scale}, we introduce the algorithm that computes an approximation to the \v{C}ech scale using the numerical bisection method.

Additionally, in Section \ref{Section:AABB} we incorporate the concept of an minimal axis-aligned bounding box (AABB) into our methodology. An AABB is a rectangular parallelepiped whose faces are perpendicular to the basis vectors. These bounding boxes frequently arise in spatial subdivision problems, such as ray tracing \cite{Mahovsky2004} and collision detection \cite{collision}. In this paper, we study AABBs to enclose the intersection of a finite collection of disks. This approach proves valuable for discerning whether the collection intersects at a singular point or not.  In this section, we also provide an algorithm for constructing the AABB of a disk system. 


\section{Intersection properties of sphere systems} \label{Section:Disk System Intersections}
Throughout this work, we will refer to a \textbf{$d$-disk system} $M$, or simply a \textbf{disk system}, as a finite collection of closed disks in $\mathbb{R}^d$ with positive and not necessarily equal radii, i.e.,

\[M=\{D_i(c_i;r_i) \subset \mathbb{R}^d \mid c_i \in \mathbb{R}^d, r_i >0, 1\leq i \leq m < \infty\}.\]

Moreover, in order to study the intersection properties of a disk system $M$ with the approach addressed in Sections \ref{Section:AABB} and \ref{Section:Rips.Cech.systems} of this work, we will conduct a study in this section of the intersection properties of the spheres corresponding to the boundaries of each disk in $M$, which we call a \textbf{sphere system} and denote by $\partial M$,

\[\partial M = \{ \partial D_i \subset \mathbb{R}^d \mid D_i \in M \},\]

where \( \partial \) denotes the topological boundary operator.

Following the notation in \cite{MAIOLI:2017}, we introduce the following generalization of the sphere.
\begin{Definition}
An $i$-sphere in $\R^d$ is the intersection of a sphere with an affine subspace of dimension $i$.
\end{Definition}

Of course, the notions of a sphere (as a $(d-1)$-dimensional surface) and a $d$-sphere in $\mathbb{R}^d$ agree. However, an $i$-sphere in $\mathbb{R}^d$ can also be viewed as the intersection of $d$-spheres. For instance, the intersection of two spheres typically occurs in a hyperplane, forming a $(d-1)$-sphere in $\mathbb{R}^d$. When another $d$-sphere intersects this configuration, the result may be a $(d-1)$-sphere, a $(d-2)$-sphere, a $0$-sphere (a single point), or it might even be empty, all within the same hyperplane. For a disk system $M = \{D_i(c_i; r_i)\}$ composed of $m$ disks, where $\{c_1, \dots, c_m\}$ is a set in general position in $\mathbb{R}^d$, the maximum dimension of the affine subspace associated with the $i$-sphere, obtained from the intersection of all the spheres in $\partial M$, is at most $d-m+1$, or equivalently, $i=d-m+1$. This conclusion is drawn from \cite[Theorem 2.1]{MAIOLI:2017} and the fact that the affine hull of $\{c_1, \dots, c_m\}$ is of dimension $m-1$. Consequently, the following result is proven.

\begin{Lemma}
Let $M = \{D_1(c_1; r_1), \ldots, D_m(c_m; r_m)\}$ be disk system such that $\{c_1, \dots, c_m\}$ is a set in general position in $\mathbb{R}^d$. Then, the possibilities for the set $\cap_{D_i \in M} \partial D_i$ are:
\begin{enumerate}
    \item the empty set;
    \item a single point;
    \item a $(d-m+1)$-sphere.
\end{enumerate}
\end{Lemma}

Remarkable points in $i$-spheres that will play a key role in the rest of the article are the \textbf{poles}.  Let $\pi_i:\R^d \longrightarrow \R$ be the canonical projection on the $i$-th factor for $i=1,...,d$, and let $\{e_1,e_2,...,e_d\}$ be the standard basis of $\R^d$.

\begin{Definition}
Let $e_q$ be the $q$-th vector of the canonical base of $\mathbb{R}^d$. An $e_q$-north (south) pole of an $i$-sphere $S$ in $\R^d$ is a point on $S$ whose projection on the $q$-th coordinate is maximum (minimum). In other words, $x\in S$ is the $e_q$-north pole if $\pi_{q}(y)\leq \pi_q(x)$ for all $y\in S-\{x\}$, where $\pi_q$ represents the projection onto the $q$-th coordinate. 

We denote the $e_q$-north pole of $S$ by $s_q^+$ and the $e_q$-south pole by $s_q^-$.
\end{Definition}

An $i$-sphere can have a single $e_q$-pole (north or south) or an infinite number of them, which occurs when a normal vector to the affine space containing the $i$-sphere is aligned with the vector $e_q$. We are interested in finding the $e_q$-poles of $(d-m+1)$-spheres originating from disk systems $M = \{D_1(c_1; r_1), \ldots, D_m(c_m; r_m)\}$, by taking the intersection $\cap_{j=1}^{m} \partial D_j$. Such $(d-m+1)$-spheres will be denoted by $S_M(c; r)$, to emphasize the disk system $M$, as well as its center and radius.

\begin{Lemma} \label{Lemma:extreme values}
    Let $M = \{D_1, \ldots, D_m\}$ be a $d$-disk system such that $\bigcap_{j = 1}^m D_j \neq \emptyset$, and let $p$ be a point in $\bigcap_{j = 1}^m D_j$ such that $\pi_q(p) \leq \pi_q(x)$ (resp. $\pi_q(p) \geq \pi_q(x)$) for every $x$ in $\bigcap_{j = 1}^m D_j$. Then, there exists an $i$-sphere $S = \partial D_{j_1} \cap \cdots \cap \partial D_{j_i}$ such that $p$ is in $S$ and $p$ is the $e_q$-south pole (resp. $e_q$-north pole) of $S$.
\end{Lemma}
\begin{proof}
    Since $\cap_{j=1}^m D_j \neq \emptyset$, then $\partial (\cap_{j=1}^m D_j) \neq \emptyset$, $\partial (\cap_{j=1}^m D_j) \subset \cap_{j=1}^m D_j$ due to the closedness of the sets $D_j$, for $j = 1,\ldots, m$, and $p \in \partial (\cap_{j=1}^m D_j)$. 

    On the other hand, since $\partial (\cap_{j=1}^m D_j) \subset \cup_{i=1}^m \partial D_j$, there exist indices $j_1, \ldots, j_i$ such that $p \in \partial D_{j_r}$ for any $r = 1, \ldots, i$; let $\Lambda(p) = \{j_1, \ldots, j_i\} \subseteq \{1, \ldots, m\}$ be a maximal subset of indices such that $p \in D_j$ if and only if $j \in \Lambda(p)$. We claim that $p \in \cap_{r = 1}^i \partial D_{j_r}$ is the $e_q$-south pole of $S := \cap_{r = 1}^i \partial D_{j_r}$. 
    
    In effect, let $V_p \subset \mathbb{R}^d$ be an open neighborhood of $p$ sufficiently small such that:
    \begin{enumerate}
        \item Every $x \in V_p \cap \partial (\cap_{j = 1}^m D_j)$ 
        has as maximal set of indices a proper subset of $\Lambda(p)$, 
        \item $S \cap V_p \subset \partial (\cap_{j = 1}^m D_j)$.
    \end{enumerate}
    The first condition can be guaranteed by the finiteness of the disk system $M$, and the second condition is a consequence of the maximality of the set $\Lambda(p)$. Therefore, $\pi_q(p) \leq \pi_q(x)$ for every $x \in S \cap V_p$, which is equivalent to the fact that $\pi_q(p) \leq \pi_q(x)$ for every $x \in S$, in the case of $i$-spheres.

\end{proof}


\subsection{Sphere systems with two spheres}
\label{Subsection:two.sphere.system}
In the following two lemmas we provide the computations to determine the center, radius and poles for a $(d-1)$-sphere given by the intersection of two disks in $\mathbb{R}^d$.

\begin{Lemma} 
\label{two-disk border poles}
Let $M = \{D_1(c_1; r_1), D_2(c_2; r_2)\}$ be a disk system with two $d$-disks such that $\partial D_1 \cap \partial D_2$ is a $(d-1)$-sphere $S=S_M(c; r)$ with center $c$ and radius $r$. Then,
\[c = \frac{1}{2}\left(1+ \frac{r_2^2-r_1^2}{\|c_2 - c_1\|^2} \right)c_1 + \frac{1}{2} \left( 1 + \frac{r_1^2 - r_2^2}{\|c_2 - c_1 \|^2} \right)c_2\] 
and
\[r = \frac{2\sqrt{s(s-\|c_2 - c_1\|)(s-r_1)(s-r_2)}}{\|c_2 - c_1\|}\]
where $s=\frac{1}{2}(\|c_2 - c_1\| + r_1 + r_2)$.
\end{Lemma}

\begin{proof}
Let $\Pi$ be the hyperplane containing the $(d-1)$-sphere $S$, which is defined by the equation:
\begin{equation}
\label{hiperplano}
\sum_{i=1}^d (k_i-h_i)x_i-\frac{1}{2}\sum_{i=1}^d (k_i^2-h_i^2)=\frac{r_1^2-r_2^2}{2},
\end{equation}
where $c_1 = (h_1, \ldots, h_d)$ and $c_2 = (k_1, \ldots, k_d)$. Then the normal vector of the hyperplane $\Pi$ is given by $N := c_2 - c_1 = (k_1 - h_1, \ldots, k_d - h_d)$, and the center $c$ of $S$ is determined by the intersection point of the hyperplane $\Pi$ with the perpendicular line that passes through the center $c_1$ of $D_1$. This line can be parameterized as $\gamma : t \mapsto c_1 + t N = (x_1(t), \ldots , x_d(t))$, such that $\gamma(0) = c_1$ and $\gamma(1) = c_2$. We can compute the intersection point $c = \gamma(t_*)$ of $\Pi$ and $\gamma([0, 1])$, for any $t_* \in [0,1]$, by substituting it in (\ref{hiperplano}),
\[ \sum_{i=1}^d (k_i-h_i)(h_i+t_*(k_i-h_i)) = \frac{1}{2}\left(\sum_{i=1}^d (k_i^2-h_i^2) + (r_1^2-r_2^2)\right). \]
And solving for $t_*$, we obtain that $t_* = \frac{1}{2} + \frac{r_1^2-r_2^2}{2\|c_2 - c_1\|^2}$. Hence, the center of $S$ is given by:
\begin{align*}
c = c_1 + t_*N = \frac{1}{2}\left(1+ \frac{r_2^2-r_1^2}{\|c_2 - c_1\|^2} \right)c_1+ \frac{1}{2}\left(1+ \frac{r_1^2-r_2^2}{\|c_2 - c_1\|^2} \right)c_2
\end{align*}
Next, we will compute the radius $r$ of $S$. This radius can be determined as the height $r$ of the triangle with base $\|c_2 - c_1\|$ formed by the points $c_1$, $c_2$, and a point on $S$. Thus, by the Heron's formula we have that
\begin{equation*} 
r = \frac{2\sqrt{s(s-\|c_2 - c_1\|)(s-r_1)(s-r_2)}}{\|c_2 - c_1\|},
\end{equation*}
where $s=\frac{1}{2}(\|c_2 - c_1\|+r_1+r_2)$ correspond to the semi-perimeter.
\end{proof}

We can proceed now to compute the poles of the $(d-1)$-sphere $\partial D_1 \cap \partial D_2$.
\begin{Lemma} 
\label{Lemma:poles of two spheres}
Let $D_1(c_1; r_1)$ and $D_2(c_2; r_2)$ be two $d$-disks such that $\partial D_1 \cap \partial D_2$ is a $(d-1)$-sphere $S = S(c; r)$ with center $c$ and radius $r$. Then, the $e_q$-poles of $S$ are $s_q^\pm = c \pm \sum_i^d x_i e_i$, where
\[x_i = 
\begin{cases}
\dfrac{r |\pi_i(c_2 - c_1) \pi_q(c_2 - c_1)|}{\|c_2 - c_1\|\sqrt{\|c_2 - c_1\|^2 - \pi_q(c_2 - c_1)^2}}, & i \neq q \\ \\
\dfrac{r{\sqrt{\|c_2 - c_1\|^2 - \pi_q(c_2 - c_1)^2}}}{\|c_2 - c_1\|}, & i = q.
\end{cases}
\]
\end{Lemma}

\begin{proof}
For simplicity, we translate the hyperplane $\Pi$, which contains the $(d-1)$-sphere $S$, as well as the sphere itself, to the origin; in such case, the corresponding equations are given by,
\begin{align*}
\sum_{i=1}^d (k_i-h_i)x_i & = 0, \\
\sum_{i=1}^d x_i^2 & = r^2,
\end{align*}
where $h_i := \pi_i(c_1)$ and $k_i := \pi_i(c_2)$ for $i =1, \ldots, d$.
In the case that $k_q - h_q = \pi_q(c_2 - c_1) = 0$, the normal vector $N = c_2 - c_1$ of the hyperplane $\Pi$ is orthogonal to the basis vector $e_q$. Therefore, the $e_q$-poles of $S$ are $c \pm re_q$, which agree with the formulae of the lemma.

On the other hand, suppose that $k_q - h_q\neq 0$. To find the $e_q$-poles of $S$, we will use the Lagrange multiplier method. Consider the following function:
\begin{equation} \label{x_q}
x_q=f(x_1,x_2,...,\widehat{x_q},...,x_{d})=\frac{-\sum_{j \neq q}^{d} (k_j - h_j)x_j}{k_q-h_q},
\end{equation}
subject to the restriction:
\label{restriction}
\[ g(x_1,x_2,...,\widehat{x_q},...,x_{d})=\sum_{j \neq q}^{d} x_j^2 + \left(\frac{-\sum_{j \neq q}^{d} (k_j - h_j)x_j}{k_q-h_q}\right)^2-r^2=0 \]
Let $\lambda$ be the Lagrange multiplier, we define 
\[ h(x_1,...,\widehat{x_q},...,x_{d},\lambda)=f(x_1,...,\widehat{x_q},...,x_{d})+ \lambda g(x_1,...,\widehat{x_q},...,x_{d}) \]
For any $i \neq q$, consider the following system of equations:
\[ \frac{\partial h}{\partial x_i}=-\frac{k_i - h_i}{k_q - h_q}+2\lambda x_i + 2\lambda \left(\frac{-\sum_{j\neq q}^d (k_j - h_j) x_j}{k_q - h_q}\right)\left(-\frac{k_i - h_i}{k_q - h_q}\right) = 0. \]
Then
\postdisplaypenalty=0
\begin{align*}
-\frac{k_i - h_i}{k_q - h_q} + 2\lambda\left( x_i+ \frac{k_i - h_i}{(k_q - h_q)^2}{\sum_{j\neq q}^d (k_j - h_j) x_j}\right) &= 0 \\
-(k_i - h_i)(k_q - h_q) + 2\lambda\left( (k_q - h_q)^2 x_i+ (k_i - h_i){\sum_{j\neq q}^d (k_j - h_j) x_j}\right) &= 0
\end{align*}
Solving this system of equations for $\lambda$, we obtain that,

\[ \lambda = \frac{(k_i - h_i)(k_q - h_q)}{2\left( (k_q - h_q)^2 x_i + (k_i - h_i){\sum_{j\neq q}^d (k_j - h_j) x_j}\right)} \]

Comparing the last expression for two indices $i \neq \tilde{i}$, we have that,

\postdisplaypenalty=0
\begin{align*}
x_i^2 &= \frac{r^2(k_i - h_i)^2 }{ (k_{\tilde{i}} - h_{\tilde{i}})^2 + \sum_{j \neq \tilde{i},q}{(k_j - h_j)}^2+ \frac{1}{(k_q - h_q)^2}\left(\sum_{j\neq q} (k_j - h_j)^2\right)^2 } \\
 &= \frac{r^2(k_i - h_i)^2 }{ \sum_{j \neq q}{(k_j - h_j)}^2+ \frac{1}{(k_q - h_q)^2}\left(\sum_{j\neq q} (k_j - h_j)^2\right)^2 } \\
&= \frac{r^2(k_i - h_i)^2 (k_q - h_q)^2 }{ \left( (k_q - h_q)^2 + \sum_{j\neq q} (k_j - h_j)^2 \right) \left(\sum_{j \neq q}{(k_j - h_j)}^2\right) } \\
&= \frac{r^2 \pi_i(c_2 - c_1)^2 \pi_q (c_2 - c_1)^2}{\|c_2 - c_1\|^2 \left(\|c_2 - c_1\|^2 - \pi_q(c_2 - c_1)^2 \right)}
\end{align*}

Finally, for $i = q$, we can use the last expression to substitute it in (\ref{x_q}) and obtain the desired result.
\end{proof}


\subsection{Sphere systems with more than two spheres}
\label{Subsection:more.sphere.system}
Now, let us proceed with the explicit calculation of the coefficients for the center $c$ of the $(d-m+1)$-sphere $S=\cap_{j=1}^{m} \partial D_j$. We can achieve this by considering the disk system translated to $c_m$, denoted as $\{D_j(c_j-c_m;r_j)\}_{j=1}^{m}$, and by defining the $(d-m+1)$-sphere $S-\{c_m\}=\cap_{j=1}^{m} \partial D_j(c_j-c_m;r_j)$. This sphere is positioned at the intersection of hyperplanes (for more details, refer to \cite{MAIOLI:2017}). 

\begin{equation}
    \label{hiperplanos}
    (c_k-c_m)^Tx=\frac{1}{2}(r_m^2+\|c_k-c_m\|^2-r_k^2)
\end{equation} 
for all $k=1,...,m-1$. 
Utilizing the information that the center of $S-\{c_m\}$ can be expressed as a combination of the centers $c_k-c_m$ and substituting it into (\ref{hiperplanos}), we obtain a linear system of equations with dimensions $(m-1)\times (m-1)$:
\begin{equation*}
    \sum_{j=1}^{m-1} \lambda_j (c_k-c_m)\cdot (c_j-c_m)=\frac{1}{2}(r_m^2+\|c_k-c_m\|^2-r_k^2) 
\end{equation*} 
for $k=1,...,m-1$. Solving the system of equations for $\lambda=(\lambda_1,...,\lambda_{m-1})$, we find the center of $S$ as follows: 
$$c=\lambda_1(c_1-c_m)+...+\lambda_{m-1}(c_{m-1}-c_m) +c_m $$
The radius of the sphere $S$ can be computed using the equation: 
$$r^2=r_k^2-\|c-c_k\|^2$$
for any $k\in \{1,...,m-1\}$.

Now that we have determined the center and radius of $S$, as well as the affine space that contains it, we can proceed to compute its $e_q$-poles for each $q\in \{1,2,...,d\}$. These poles reside in the affine space that contains $S$ and within a set that we define below.

Let $S$ be an $i$-sphere in $\R^d$, and let $n_1,...,n_{d-i}$ be orthogonal vectors to the affine space $L$ that contains $S$. Consider the space $M$ generated by these vectors together with the vector $e_q$ from the canonical basis of $\R^d$. Let us denote $n_j=\left(n_l^{(j)}\right)_{l=1}^d$ for each $j=1,...,d-i$. Then, we can define $L_0$, the set $L$ translated to the origin, as follows:
\begin{align*}
L_0&=\left<\{n_j\}_{j=1}^{d-i} \right>^{\perp}\\
&=\left\{ x\in \R^d \mid x \cdot n_j=0, \hspace{.2cm} \forall j=1,...,d-i \right\}
\end{align*}
The set $M$ is defined as:
\begin{align*}
M&=\left<\{n_j\}_{j=1}^{d-i}\cup \{e_q\}  \right>\\
&=\left\{ x=\sum_{j=1}^{d-i}\lambda_j n_j+\lambda_{d-i+1}e_q\in \R^d \mid\lambda_j\in \R, \hspace{.2cm} \forall j=1,...,d-i+1 \right\}\\
&= \left\{ \left(\sum_{j=1}^{d-i} \lambda_j n_1^{(j)},..., \sum_{j=1}^{d-i} \lambda_j n_q^{(j)}+\lambda_{d-i+1},...,\sum_{j=1}^{d-i} \lambda_j n_d^{(j)} \right) \mid\lambda_j\in \R \right\} 
\end{align*} 
Refer to Figure \ref{intergeo4} for a visual representation of the subspaces $L$ and $M+\{c\}$.
\begin{figure}[htb]
    \centering 
    \includegraphics[width=.4\textwidth]{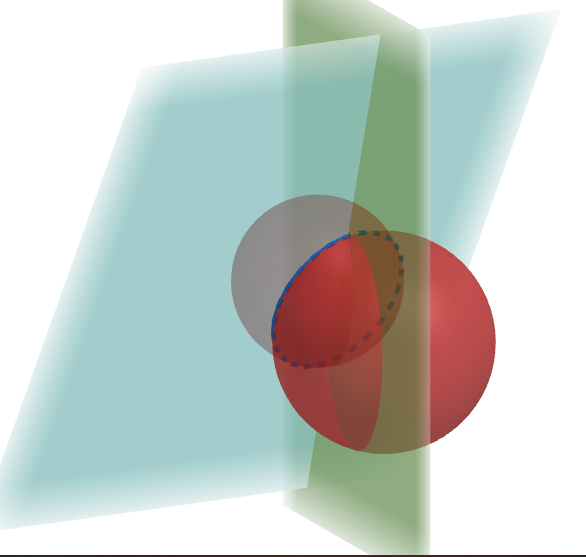}
    \caption{Visualization of the subspaces $L$ and $M+\{c\}$.} \label{intergeo4}
\end{figure}

As mentioned above, the $e_q$-poles of $S$ lie at the intersection of $L$ and $M+\{c\}$, where $c$ is the center of $S$. To simplify the calculations, we will utilize $L_0$ and $M$, and then translate them into $c$. The intersection of $M$ and $L_0$ can be expressed as follows:
{\scriptsize
\begin{align*}
    M\cap L_0&=\left\{x\in M\mid  x\cdot n_k=0,  \hspace{.2cm}\forall k=1,...,d-i\right\} \\
    &=\left\{\sum_{j=1}^{d-i}\lambda_j n_j+\lambda_{d-i+1}e_q\in \R^d \left| \left(\sum_{j=1}^{d-i}\lambda_j n_j+\lambda_{d-i+1}e_q\right) \cdot n_k=0, \hspace{.2cm} \lambda_j\in \R  \hspace{.2cm}\forall k=1,...,d-i \right. \right\} \\
    &=\left\{\sum_{j=1}^{d-i}\lambda_j n_j+\lambda_{d-i+1}e_q\in \R^d \left| \sum_{j=1}^{d-i}\lambda_j n_j\cdot n_k+\lambda_{d-i+1}n_q^k =0, \hspace{.2cm} \lambda_j\in \R  \hspace{.2cm}\forall k=1,...,d-i \right. \right\} 
\end{align*} 
}

Let us consider a disk system in $\mathbb{R}^d$, denoted $\{D_j(c_j;r_j)\}_{j=1}^m$, where $m<d$. The intersection of their boundaries forms a $(d-m+1)$-sphere $S$. In this case, the subspace $M$ has dimension $m$, or $dim(M) = m-1$ if $e_q \in M$. We choose the normal vectors for the affine space containing $S$ as $n_j=c_j-c_m$, where $j=1,...,m-1$. Then
{\scriptsize
\begin{equation*}
    M= \left\{ \left(\sum_{j=1}^{m-1} \lambda_j \left(c_1^{(j)}-c_1^{(m)}\right),..., \sum_{j=1}^{m-1} \lambda_j \left(c_q^{(j)}-c_q^{(m)}\right)+\lambda_{m},...,\sum_{j=1}^{m-1} \lambda_j \left(c_d^{(j)}-c_d^{(m)}\right) \right) \left| \lambda_j\in \R \right. \right\}
\end{equation*} 
}
By rewriting, we have
{\scriptsize
\begin{align*}
    M\cap L_0
    &=\left\{\sum_{j=1}^{m-1}\lambda_j n_j+\lambda_{m}e_q\in \R^d \left|  \sum_{j=1}^{m-1}\lambda_j n_j\cdot n_k+\lambda_{m}n_q^k =0, \hspace{.2cm} \lambda_j\in \R,  \hspace{.2cm}\forall k=1,...,m-1\right.\right\} 
\end{align*}
}

If $S(c;r)=\cap_{j=1}^m \partial D_j$ is the $(d-m+1)$-sphere with center in $c$ and radius $r$, then the $e_q$-poles of $S$ are the $e_q$-poles of $S-\{c_m\}$ but translated by $c$. The poles of $S-\{c_m\}$ are located in $M\cap L_0$. If $p$ is an $e_q$-pole of $S-\{c_m\}$, then it can be expressed as 
$$p=\sum_{j=1}^{m-1}\lambda_j n_j + e_q\lambda_m$$ 
 for some $ \lambda_j \in \R$, $j=1,...,m$ and the following conditions holds:  
$$\sum_{j=1}^{m-1} \lambda_j n_j\cdot n_k + \lambda _m n_q^k=0$$ 
for each $k=1,...,m-1$ and 
$$\|p\|^2=r^2$$

Thus, if $p=\sum_{j=1}^{m-1}\lambda_j n_j + e_q\lambda_m$ is an $e_q$-pole of $S-\{c_m\}$, the following equations are satisfied for $\lambda_1,\lambda_2,...,\lambda_m \in \R$:
\begin{equation}
    \label{lineal}
    \sum_{j=1}^{m-1} \lambda_j n_j\cdot n_k + \lambda _m n_q^k=0
\end{equation}

\begin{equation}
    \label{cuadratica}
    \sum_{i=1}^d \left(\sum_{j=1}^{m-1} \lambda_j n_i^j \right)^2 + 2\lambda_m \sum_{j=1}^{m-1}\lambda_j n_q^j + \lambda_m^2=r^2
\end{equation} 
for all $k=1,...,m-1$, with $r$ the radius of the $(d-m+1)$-sphere $S$.
From (\ref{lineal}) we have the system
\begin{equation*}
\begin{pmatrix}
    n_1\cdot n_1 & n_1\cdot n_2 & .... &n_1\cdot n_{m-1} \\
    n_2\cdot n_1 & n_2\cdot n_2 &.... &n_2\cdot n_{m-1} \\
    ...\\
    n_{m-1}\cdot n_1 & n_{m-1}\cdot n_2 &.... &n_{m-1}\cdot n_{m-1} \\
\end{pmatrix} \begin{pmatrix}
    \lambda_1 \\
    \lambda_2\\
    ...\\
    \lambda_{m-1}
\end{pmatrix} + \lambda_m
\begin{pmatrix}
    n_q^1 \\
    n_q^2\\
    ...\\
    n_q^{m-1}
\end{pmatrix} =0.
\end{equation*}

Let us denote $A$ as the matrix $(n_i\cdot n_j)_{i,j}$ and $B$ as the vector $(-n_q^j)_{j=1}^{m-1}$. Then, we have $A\lambda=\lambda_m B$, where $\lambda=(\lambda_1,...,\lambda_{m-1})$. Solving for $\lambda$, we obtain $\lambda_{j}=\lambda_m (A^{-1}B)[j]$ for each $j=1,...,m-1$ (where $(A^{-1}B)[j]$ denotes the entry $j$ of the $(m-1)\times 1$ vector $(A^{-1}B)$). By substituting the value of $\lambda_j$ into (\ref{cuadratica}), we obtain the quadratic equation:
$$\lambda_m^2\left[\sum_{i=1}^d\left( \sum_{j=1}^{m-1} (A^{-1}B)[j]n_{i}^j\right)^2 + 2\sum_{j=1}^{m-1} (A^{-1}B)[j]n_{q}^j +1\right]-r^2=0$$
Let us define 
$$\Gamma_i=\sum_{j=1}^{m-1}(A^{-1}B)[j]n_{i}^j$$
for all $i=1,. ..,d$. Solving this equation, we find: 

$$\lambda_m=\frac{\pm r}{\sqrt{\sum_{i=1}^d \Gamma_i^2 +2 \Gamma_q+1}}$$  $$\lambda_j=\frac{\pm r (A^{-1}B)[j]}{\sqrt{\sum_{i=1}^d \Gamma_i^2 +2 \Gamma_q+1}} $$
for $j=1,...,m-1$. Therefore, the $e_q$-poles of $S$, for $q=1,...,d$, are:
\begin{equation*}
    p=\sum_{j=1}^{m-1}\lambda_j n_j + e_q\lambda_m +c =\frac{\pm r }{\sqrt{\sum_{i=1}^d \Gamma_i^2 +2 \Gamma_q+1}}\left(\sum_{j=1}^{m-1}(A^{-1}B)[j] n_j +e_q \right)+c
\end{equation*}

\section{Vietoris-Rips and \texorpdfstring{\v{C}}{C}ech systems} \label{Section:Rips.Cech.systems}

Our goal in this section is to provide a comprehensive understanding of the disk system, the Vietoris-Rips system, and the \v{C}ech system. Additionally, we introduce some results that establish a certain connection between both disk systems. Investigating the features and qualities of data and spaces can provide us with useful knowledge about their geometric and topological characteristics.

Before we look into the definitions of Vietoris-Rips and \v{C}ech systems, let us give a brief overview. These systems are essential in the field of topological data analysis for recognizing and comprehending the geometric structure of point cloud data. Vietoris-Rips complex and \v{C}ech complex share the goal of capturing the topology of the underlying metric space, both provide different ways of recognizing connections and associations among data points. The Vietoris-Rips complex tends to be more efficient and scalable for large datasets, while the \v{C}ech complex can be more accurate but computationally more expensive. The choice between the two depends on the nature of the dataset and the specific goals of the topological analysis. Now, let us move on to defining these fundamental concepts.

\begin{Definition} \label{Definition:RipsCechSystems}
Let $M=\{D_1,D_2,\ldots , D_m \}$ be a $d$-disk system. We say $M$ is a \textit{Vietoris-Rips system} if $D_i\cap D_j \neq \emptyset$ for each pair $i,j\in \{1,2,\ldots, m\}$. Furthermore, if the $d$-disk system $M$ has the nonempty intersection property $\bigcap_{D_i \in M} D_i \neq \emptyset$, then $M$ is called a \textit{\v{C}ech system}.
\end{Definition}

For each $\lambda \geq 0$, we define a collection of $d$-disks $M_\lambda$ with the same centers as those in the $d$-disk system $M$, but with radii rescaled by $\lambda$. When $\lambda > 0$, $M_\lambda$ is a $d$-disk system again. $M_1$ is equal to $M$, and $M_0$ is the set of the centers of the $d$-disks in $M$.

In the field of topological data analysis, the Rips scale and the \v{C}ech scale are essential parameters for determining the closeness and connectivity between data points. These two scales offer different perspectives on how we measure and comprehend geometric relationships within point-cloud data. To understand their importance in capturing the underlying topological structure, let us look at their definitions. The Vietoris-Rips scale $\nu_M$, of a $d$-disk system $M$ is the smallest $\lambda \in \mathbb{R}$ such that $M_\lambda$ is a Vietoris-Rips system. Similarly, the \v{C}ech scale $\mu_M$, of $M$ is the smallest $\lambda \in \mathbb{R}$ such that $M_\lambda$ is a \v{C}ech system. This is 
    \begin{align*}
        \nu_M&=\inf \{\lambda\in \R \mid M_{\lambda} \mbox{ is a Vietoris-Rips system}\}\\
        \mu_M&=\inf \{\lambda\in \R \mid M_{\lambda} \mbox{ is a \v{C}ech system}\}
    \end{align*}

Next, we present some easily observable properties for both scales. It can be easily seen that $M$ is a Vietoris-Rips system if and only if $\nu_M \leq 1$ (in particular $\nu_{M_{\nu_M}} = 1$); similarly, $M$ is a \v{C}ech system if and only if $\mu_M \leq 1$.

Note that for a given $d$-disk system $M=\{D_1,D_2,\ldots , D_m \}$ the Vietoris-Rips scale is 
$$\nu_M = \max_{i<j}\{\Vert c_i - c_j \Vert/(r_i + r_j)\}$$
where $c_i$ and $r_i$ are the center and radii of $D_i$. An additional observation is that, in cases where the disk system consists of either one or two disks, the Vietoris-Rips scale coincides with the \v{C}ech scale.  It is evident that every \v{C}ech system is also a Vietoris-Rips system; however, the reverse assertion, in general, is not true.

Conversely, if the $d$-disk system contains at least three disks, determining the \v{C}ech scale becomes  more complex.
 In the context of \v{C}ech scale, the following remark is important and  play a key role in implementation (see \cite{Espinoza2019} for details).

\begin{Remark} \label{Remark:unique-point}
\normalfont
If $\mu_M$ is the \v{C}ech scale for $M$, then the $\mu_M$-rescaled system $M_{\mu_M}$, has only one point in the intersection $\bigcap_{D_i \in M} D_i(c_i; \mu_M r_i)$. 
\end{Remark}

As we have mentioned, a \v{C}ech system is also a Vietoris-Rips system, but the converse is not true. What we can affirm is that if a system is a Vietoris-Rips system, then the system rescaled by the factor $\sqrt{2d/(d+1)}$ is also a \v{C}ech system. This is established by the following lemma, the proof of which can be found in \cite{Espinoza2019}.

\begin{Lemma} \label{Lemma:VR-generalized}
Let $M = \{D_i(c_i; r_i)\}$ be a $d$-disk system in euclidean space $\mathbb{R}^d$. If $D_i(c_i; r_i) \cap D_j(c_j; r_j) \neq \emptyset$ for every pair of disks in $M$, then 
\[ \bigcap_{D_i \in M} D_i(c_i; \sqrt{2d/(d+1)}\, r_i) \neq \emptyset. \]
\end{Lemma}

One of the implications of the previous result is that, for any given disk system $M$, we can bound the \v{C}ech scale using the Vietoris-Rips scale $\nu_M$. This is stated by the following corollary.

\begin{Corollary}\label{Corollary:Cech-system}
If $M$ is an arbitrary $d$-disk system and $\nu_M$ is its Vietoris-Rips scale, then its \v{C}ech scale satisfies $\nu_M\leq \mu_M \leq\sqrt{2d/(d+1)}\, \nu_M$. Therefore, for every $d$-disk system $M$, the rescaled disk system $M_{\sqrt{2d/(d+1)}\, \nu_M}$ is always a \v{C}ech system.
In particular, if $\sqrt{2d/(d+1)}\, \nu_M \leq 1$ then $M_{\nu_M}$ is a \v{C}ech system. 
\end{Corollary}


\subsection{Algorithm for determining \v{C}ech system.}
\label{Subsection:Cech.system}

In the previous section, we have determined the $e_q$-poles for the intersection of any number of disks in $\R^d$. If any of these poles is in all disks of the disk system, it indicates that the system conforms to the criteria of a \v{C}ech system. It is important to recognize that this result streamlines our calculation process, focusing on specific points to establish whether the system exhibits a non-empty intersection.

Given a system of $m$ disks in $\mathbb{R}^d$ where $m>d$, it is enough to verify if every subsystem of $d+1$ disks qualifies as a \v{C}ech system to conclude that the entire system of disks has a non-empty intersection. This assertion is supported by the Helly's Theorem.

Now, we introduce an algorithm that determines whether a disk system qualifies as a \v{C}ech system. In simpler terms, if the disk system exhibits a non-empty intersection, the algorithm outputs "TRUE"; otherwise, it outputs "FALSE". The algorithm operates by seeking poles within the intersections of the disk boundaries, which, as we have observed, correspond to $i$-spheres. It initiates the search for poles within individual disks and then progresses to the pairwise intersections of the disk boundaries ($(d-2)-sphere$), continuing the process iteratively. If a pole is found within the remaining disks, the system is classified as a \v{C}ech system.

\begin{algorithm}
\caption{\texttt{Cech.system}}
\label{Algorithm:intersection_property}
    \SetKwInOut{Input}{Input}
    \SetKwInOut{Output}{Output}

    \Input{A $d$-disk system $M=\{D_j\}_{j=1}^m$}
    \Output{A logical \texttt{TRUE}/\texttt{FALSE} to indicate if $M$ is a \v{C}ech system}
    Initialize: $\texttt{Is\_Cech\_System} \leftarrow \texttt{FALSE}$\\
    \For{$k \leftarrow 1${\normalfont\textbf{ to }}$m$}{
        Let $\mathcal{S}$ be the set of $(d-k+1)$-spheres of $\partial M$\\
        \For{$S${\normalfont\textbf{ in }}$\mathcal{S}$}{
            \For{$q \leftarrow 1${\normalfont\textbf{ to }}$d$}{
                Compute the set $\{s_q^\pm\}$ of $e_q$-poles of $S$\\
                \For{$s \leftarrow \{s_q^\pm\}$}{
                    \If{$s \in \cap_{j=1}^{d} D_j$}{
                        $\texttt{Is\_Cech\_System} \leftarrow \texttt{TRUE}$\\
                        Go to line 16
                    }
                }
            }
        }
     }
     \Return($\texttt{Is\_Cech\_System}$)
\end{algorithm}

\begin{theorem}
Let $M = \{D_1, \ldots, D_m\}$ be a $d$-disk system. Then, $M$ is a \v{C}ech system if and only if \textup{\texttt{Cech.system$(M) = $ TRUE}}.
\end{theorem}
\begin{proof}
If \textup{\texttt{Cech.system$(M) = $ TRUE}}, then the \textup{\texttt{Cech.system}} algorithm (Algorithm \ref{Algorithm:intersection_property}) found a pole contained in the intersection $\cap_{j=1}^m D_j$, therefore $\cap_{j=1}^m D_j \neq \emptyset$ and it follows that $M$ is a \v{C}ech system.

On the other hand, if $\bigcap_{j=1}^m D_j \neq \emptyset$, let $p$ be a point in $\bigcap_{j=1}^m D_j$ satisfying $\pi_1(p) \leq \pi_1(x)$ for every $x$ in $\bigcap_{j=1}^m D_j$. By Lemma \ref{Lemma:extreme values}, it follows that $p$ belongs to an $i$-sphere and must be an $e_1$-south pole. Therefore, by the exhaustive search of Algorithm \ref{Algorithm:intersection_property} across all poles, its output is \textup{\texttt{Cech.system$(M) = $ TRUE}}.
\end{proof}


\subsection{Algorithm to compute the \texorpdfstring{\v{C}}{C}ech scale.} \label{Section:cech.scale}

Finding the minimum parameter for which the rescaled system of disks has a non-empty intersection is significant because it helps identify a critical threshold at which the disks come into contact. This parameter, known as the \v{C}ech scale, provides valuable information about the proximity or overlap of the disks, which can be crucial in various applications such as collision detection in computer graphics, spatial packing problems, and modeling physical phenomena. In this section, we introduce an algorithm to compute an approximation of the \v{C}ech scale for a system of $m$ disks in $\mathbb{R}^d$.
\begin{algorithm}
    \caption{\texttt{Cech.scale}}
    \label{Algorithm:cech.scale}
    \SetKwInOut{Input}{Input}
    \SetKwInOut{Output}{Output}
    
    \Input{A $d$-disk system $M$ in $\mathbb{R}^d$ and precision parameter $\eta > 0$}
    \Output{$\mu_M$, a \v{C}ech scale approximation}
    Compute the Vietoris-Rips scale of $M$: $\nu_M$\\
    \eIf{\textup{\texttt{Cech.system}}$(M_{\nu_M})=$ \textup{TRUE}}{
        $\mu_M \leftarrow \nu_M$\\
        Go to line 16 \\
    }{
        Initialize: $\mu_M^* \leftarrow \nu_M$, $\mu_M \leftarrow \sqrt{2d/(d+1)}\nu_M$\\
        \While{$\mu_M - \mu_M^* > \eta$}{
            Compute: $\lambda \leftarrow \dfrac{\mu_M^* + \mu_M}{2}$\\
            \eIf{\textup{\texttt{Cech.system}}$(M_{\lambda})=$ \textup{TRUE}}{
                Update: $\mu_M \leftarrow \lambda$
                }{
                Update: $\mu_M^* \leftarrow \lambda$}
        }
    }
    \Return($\mu_M$)
\end{algorithm}

    

The given code presents an algorithm to compute an approximation of the \v{C}ech scale of a disk system in Euclidean space using Algorithm \ref{Algorithm:intersection_property} and a precision parameter $\eta>0$. It initializes the scale factor $\lambda$ to the Rips scale $\nu_M$. If this scale satisfies \texttt{Cech.system}$(M_{\nu_M})=$ \textup{TRUE}, it indicates that the \v{C}ech scale has been found. Otherwise, we initiate a cycle in which we compute \texttt{Cech.system} of the system rescaled by a factor $\lambda$. The \v{C}ech scale is known to fall between the Rips scale and the value \(\sqrt{{2d}/{(d+1)}}\nu_M\) (Generalized Vietoris-Rips, Corollary \ref{Corollary:Cech-system}). To approximate the \v{C}ech scale, we employ the bisection method as long as the interval enclosing the \v{C}ech scale has a length greater than $\eta$. Finally, the algorithm returns an approximation of the \v{C}ech scale.

Utilizing the previously described algorithm, we can construct the filtered generalized \v{C}ech complex for a disk system \(M\). Let \(\mathscr{C}(M)\) denote the set of \v{C}ech subsystems, and \(\mathscr{C}_{\lambda}(M)\) the set of \v{C}ech subsystems for the rescaled disk system \(M_{\lambda}\). The \v{C}ech filtration of the \(M\) system forms a maximal chain of \v{C}ech complexes 
$$\mathscr{C}_*(M): \mathscr{C}_0(M) \subsetneq \mathscr{C}_{\lambda_1}(M) \subsetneq \mathscr{C}_{\lambda_2}(M) \subsetneq... \subsetneq \mathscr{C}_{\mu_M}(M),$$
where each \(\lambda_i\) represents the \v{C}ech scale of the system \(M_{\lambda_i}\). Since the \v{C}ech scale of a disk system indicates the factor by which we must rescale the system to make it \v{C}ech, defining a level of the filtration, \(\mathscr{C}_{\lambda}(M)\), simply requires determining the \v{C}ech scale of the system \(M_{\lambda}\).


\section{Minimal Axis-Aligned Bounding Box.} \label{Section:AABB}

In this section, we introduce the concept of the minimal axis-aligned bounding box (AABB) for the intersection of $d$-disks and present methods for its computation. The AABB provides a simplified representation of the disk intersection, making it easier to obtain valuable information about the disks. This information could be useful for computing the \v{C}ech scale of a disk system.

\begin{Definition}
Let $M$ be a disk system in $\mathbb{R}^d$. The minimal axis-aligned bounding box of $M$, denoted as $AABB(M)$, is defined as the smallest axis-aligned bounding box that contains the intersection
   $D=\cap_{D_i\in M} D_i$
, given by
 $$AABB(M):=\bigcap_{D\subset \tilde{B}} \tilde{B} $$
where $\tilde{B}$ ranges over all axis-aligned bounding boxes that contain $D$.
\end{Definition}

Note that the AABB can be expressed as 
\[ AABB(M)= \prod_{k=1}^d [\inf \pi_k(\partial D),\sup \pi_k(\partial D)] \]
where $\pi_k:\mathbb{R}^d \longrightarrow \mathbb{R}$ is the canonical projection onto the $k$-th factor, and $\partial D$ denotes the boundary of $D$. In other words, the AABB of a disk system $M$ is given by the Cartesian product of intervals, where each interval is determined by the minimum and maximum values of the corresponding projection of the disk boundaries.


\subsection{Minimal axis-aligned bounding box for two disks}
Let's consider the situation when the disk system $M$ is composed of two disks $D_1$ and $D_2$ in $\mathbb{R}^d$. If $D_1 \cap D_2 \neq \emptyset$ and $D_1 \neq D_2$, the subset $\partial D_1 \cap \partial D_2$ can take one of three forms: an empty set, a single common point (when the disks are tangent), or a $(d-2)$-dimensional sphere. In the last case, we denote the $(d-2)$-dimensional sphere or the $(d-1)$-sphere $\partial D_1\cap \partial D_2$ by $S_{1,2}$. To calculate $\inf \pi_i(\partial (D_1\cap D_2))$ and $\sup \pi_i(\partial (D_1\cap D_2))$ for each $i\in \{1,2,...,d\}$, we can use:

\begin{align}
\label{two_disks_intersection}
\inf\pi_i(\partial (D_1 \cap D_2)) = \begin{cases}
\pi_i(c_1-r_1 e_i) & \text{if } c_1-r_1 e_i \in D_2, \\
\pi_i(c_2-r_2 e_i) & \text{if } c_2-r_2 e_i \in D_1, \\
\inf(\pi_i(\partial D_1 \cap \partial D_2)) & \text{otherwise},
\end{cases}
\end{align}

\begin{align*}
\sup\pi_i(\partial (D_1 \cap D_2)) = \begin{cases}
\pi_i(c_1+r_1 e_i) & \text{if } c_1+r_1 e_i \in D_2, \\
\pi_i(c_2+r_2 e_i) & \text{if } c_2+r_2 e_i \in D_1, \\
\sup(\pi_i(\partial D_1 \cap \partial D_2)) & \text{otherwise}.
\end{cases}
\end{align*}

Indeed, by Lemma \ref{Lemma:extreme values}, the extremes of the AABB are the projections of certain poles, either from the $(d-1)$-sphere or some \(d\)-sphere. The \(d\)-spheres represent the boundaries of each disk, with poles given by \(c_j \pm r_j e_i\), and the \((d-1)\)-sphere is \(S_{1,2} = \partial D_1 \cap \partial D_2\), whose poles are computed using Lemma \ref{Lemma:poles of two spheres}. It is worth noting that there are no further options for \((d-m+1)\)-spheres in the case of a two-disk system.

To simplify the notation, we will use $B_{i,j}$ to denote the AABB of the intersection of disks $D_i$ and $D_j$. Figure \ref{2disks} illustrates the AABB of the intersection of two disks in the plane. 

\begin{figure}[htb]
    \centering 
    \includegraphics[width=.4\textwidth]{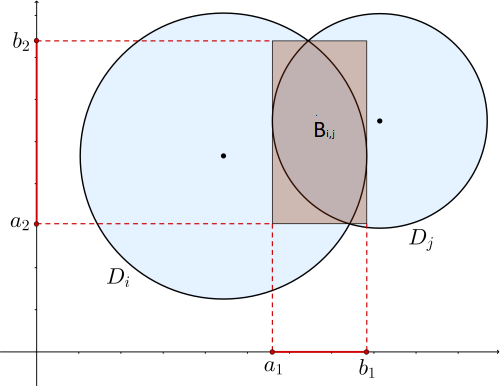}
    \caption{AABB of two disks.} \label{2disks}
\end{figure}

According to (\ref{two_disks_intersection}) and Lemma \ref{two-disk border poles}, we have a method to calculate the axis-aligned boundary box (AABB) for systems of two disks.

Knowing how to compute the AABB of two disks is not sufficient to determine the AABB for a disk system with more than two disks in $\mathbb{R}^d$. In the following examples, we demonstrate that the AABB of a disk system is not simply the intersection of all AABBs of two disks in $\mathbb{R}^3$.

\begin{example}
\label{interseccion de AABB por pares no es la AABB}
Let 
\[ M = \{D_1((4,1,0);\sqrt{2}), D_2((4,-1,0);\sqrt{2}), D_3((0,0,0);3)\} \subset \mathbb{R}^3 \]
be a Vietoris-Rips system in $\mathbb{R}^3$ with the following projection onto the $xy$-plane:

\begin{figure}[htb]
    \centering 
    \includegraphics[width=.4\textwidth]{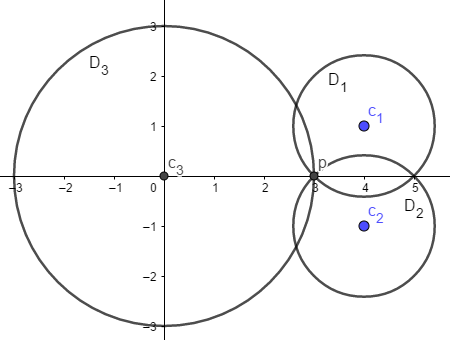}
    \caption{Disk system $M$ projected onto the $xy$-plane.} 
\end{figure}

By computing the boxes $B_{i,j}$ for all $i<j$, we obtain:
\begin{align*}
B_{1,2} &= [3,5] \times [-\sqrt{2}, \sqrt{2}] \times [-1,1] \\
B_{1,3} &= [4-\sqrt{2},3] \times [0, 1.41] \times [-0.72,0.72] \\ 
B_{2,3} &= [4-\sqrt{2},3] \times [-1.4, 0] \times [-0.72,0.72] 
\end{align*}
Therefore, $\cap_{i<j} B_{i,j} = [3,3] \times [0,0] \times [-0.72,0.72]$. However, the disk system intersects at the point $P=(3,0,0)$, which means $AABB(M) = \{P\}$. In other words, $AABB(M)$ is not equal to $\cap_{i<j} B_{i,j}$.
\end{example}

We know  that if $D \neq \emptyset$, then $\cap_{i<j} B_{i,j} \neq \emptyset$. However, the converse is not always true. The following example illustrates this fact.

\begin{example} 
\label{interseccion de AABB por pares no vacia y D vacio}
Let $N=\{D_1,D_2,D_3((0,1,0);\sqrt{10}),D_4((3,0,1);0.9)\} \subset \mathbb{R}^3$ be a Vietoris-Rips system. The projection of disks $D_1, D_2,$ and $D_3$ onto the $xy$-plane is illustrated in Figure \ref{fig:disk_system}.

\begin{figure}[htb]
    \centering 
    \includegraphics[width=.4\textwidth]{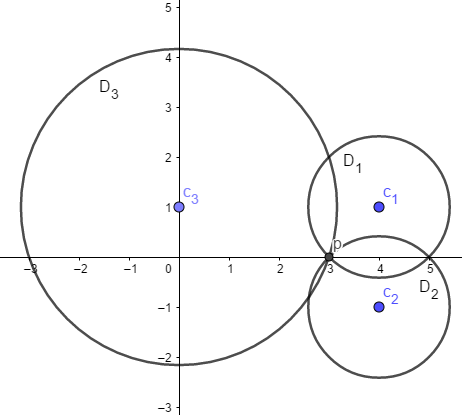}
    \caption{Disk system $N$ projected onto the $xy$-plane.} 
    \label{fig:disk_system}
\end{figure}

We will now compute the intersections $B_{i,j}$ for different pairs of disks:
\begin{align*}
B_{1,2} &= [3,5]\times [-\sqrt{2}, \sqrt{2}] \times [-1,1] \\
B_{1,3} &= [4-\sqrt{2},\sqrt{10}]\times [0, 2] \times [-1,1] \\ 
B_{1,4} &= [2.6,3.9]\times [-0.4, 0.9] \times [0.100007,1.3] \\
B_{2,3} &= [2.6,3]\times [-0.8, 0] \times [-0.44,4.47] \\ 
B_{2,4} &= [2.6,3.9]\times [-0.9, 0.4] \times [0.100007,1.3] \\ 
B_{3,4} &= [2.1,3.11]\times [-0.73, 0.9] \times [0.1,1.7]  
\end{align*}

The intersection of all pairwise intersections, $\cap_{i<j} B_{i,j}$, is given by $[3,3] \times [0,0] \times [0.1,1.7]$. However, the intersection of disks $D_1, D_2,$ and $D_3$ is a single point $P$, which is not contained in $D_4$ (by construction). Therefore, $D$ is an empty set, but $\cap_{i<j} B_{i,j}$ is not.
\end{example}

These examples clearly illustrate that when dealing with the AABB of three disks or more, knowing the AABB for pairs of disks is insufficient. In Example \ref{interseccion de AABB por pares no es la AABB}, we observe that the intersection of $AABB(\{D_i,D_j\})$ is not equal to the AABB of the disk system $M$. Similarly, in Example \ref{interseccion de AABB por pares no vacia y D vacio}, we find that the intersection of three disks is empty, yet the intersection of $AABB(\{D_i,D_j\})$ contains points.

Therefore, the next crucial step is to determine how to calculate the AABB of a disk system consisting of more than two disks in $\mathbb{R}^d$.


\subsection{Minimal axis-aligned bounding box for more than two disks}
Given a system $M$ consisting of $m$ disks in $\mathbb{R}^d$, we can compute the $e_q$-poles for any subcollection of disks. Using these $e_q$-poles, we can determine the axis-aligned bounding box (AABB) of $M$.

If $m \leq d$, we calculate the $e_q$-poles of the $(d-m+1)$-sphere $\cap \partial D_i$, and with these poles, we define the AABB of $M$ by taking $\inf \pi_q \partial D = \pi_q(p)$, where $p$ is the $e_q$-south pole of the $(d-m+1)$-sphere (similarly for $\sup \pi_q \partial D$).

In the case where $m=d+1$, we consider $AABB(M(1)),...,AABB(M(m))$ as a collection of minimal axis-aligned boxes for the disk system $M(i)=M-\{D_i\}$ in $\mathbb{R}^d$. If the intersection of any $d+1$ of these sets is nonempty, then the intersection of the entire collection gives us the minimal axis-aligned box of the disk system $M$. This can be expressed as $AABB(M)=\cap AABB(M(i))$. The next theorem confirms this finding.

\begin{theorem}[Helly's theorem for minimal axis-aligned boxes]
\label{Helly}
Let $M$ be a Rips system with $d+1$ disks in $\mathbb{R}^d$. Then, the minimal axis-aligned bounding box for the intersection set $D=\cap_{j=1}^{d+1} D_j$ satisfies:
\[ AABB(M) = \bigcap_{j=1}^{d+1} AABB (M(j)) \]
where $M(j)=M-\{D_j\}$.
\end{theorem}

\begin{proof}
We know that $AABB(M)\subseteq \bigcap_{j=1}^{d+1} AABB (M(j))$. Now, our objective is to establish the reverse inclusion, that is, $\bigcap_{j=1}^{d+1} AABB (M(j)) \subseteq AABB(M)$. In order to derive a contradiction, suppose that the reverse inclusion is not true. By definition, we have:
 \[ AABB(M)=\prod_{i=1}^d \left[\inf\pi_i(\partial D), \sup\pi_i (\partial D) \right] \]
and 
\[ \bigcap_{j=1}^{d+1} AABB (M(j)) = \prod_{i=1}^d \left[\max_k \{\inf\pi_i (\partial \cap_{j\neq k} D_j)\} , \min_k \{\sup\pi_i (\partial \cap_{j\neq k} D_j)  \}\right] \]

Without loss of generality, let's assume that:
\[ \pi_1(p)> \max_k\left\{\inf \pi_1(\cap_{j\neq k}\partial D_j)\right\}_{k=1}^{d+1} \]
where $p\in D$ satisfies $\pi_1(p)=\inf\pi_1(\partial D)$.

Let $q_j$ be the point in $\cap_{k\neq j} D_k$ such that $\pi_1(q_j)=\inf \pi_1(\cap_{k\neq j} D_k)$ for each $j=1,2,...,d+1$. Note that $q_j\notin D_j$ because $q_j$ is not in $D$.  
Now, let $\gamma_j$ be the line segment that connects $q_j$ and $p$.

Choose $\epsilon >0$ small enough such that the hyperplane $P: x_1 = \pi_1(p)-\epsilon$ does not contain any $q_j$, and  $\pi_1(q_j) < \pi_1(p)-\epsilon$ for all $j =1,2,...,d+1$ (such  hyperplane $P$ exists because $\pi_1(p)> \pi_1(q_j)$ for all $j=1,...,d+1$). 
Since $q_j$ is in $\cap_{k\neq j} D_k$, the hyperplane $P$ intersects every disk in $M(j)$ for each $j$, and $\gamma_j \subset \cap_{k\neq j} D_k$ intersects $P$ at a point that is in $\cap_{k\neq j} D_k$ (see Figure \ref{hyperplane}). Furthermore, $D_k \cap P$ is a $(d-1)$-dimensional disk. Therefore, we have a collection $\mathcal{D} = \{D_j\cap P | j=1,2,...,d+1\}$ of $d+1$ disks, each of dimension $d-1$, such that every subset $A$ of $\mathcal{D}$ consisting of $d$ disks has the non-empty intersection property. By Helly's Theorem, the intersection of all $(d-1)$-disks in $\mathcal{D}$ is not empty. Therefore, there exists a point $q\in D$ with $\pi_1(q)<\pi_1(p)$. However, this contradicts the fact that $p \in D$ is such that $\pi_1(p)= \inf \pi_1(\partial D)$.

Therefore, we conclude that $\bigcap_{j=1}^{d+1} AABB (M(j)) = AABB(M)$.
\begin{figure}[htb]
    \centering 
    \includegraphics[width=.4\textwidth]{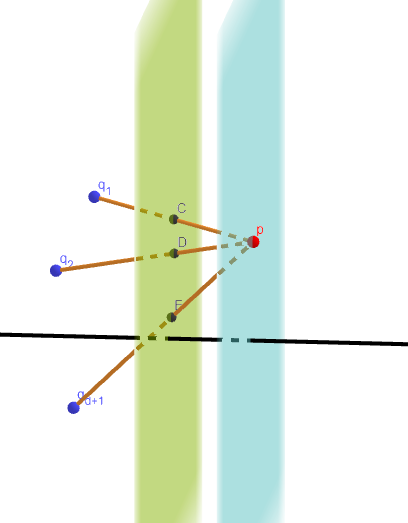}
    \caption{}
    \label{hyperplane}
\end{figure}
\end{proof}

\begin{Lemma}
Let $M$ be a Vietoris-Rips system of $d+1$ disks in $\mathbb{R}^d$. If $D=\cap_{i=1}^{d+1} D_i =\emptyset$ and $AABB(M(j)) \neq \emptyset$ for each $j=1,...,d+1$, then $\cap_{i=1}^{d+1} AABB(M(i))$ consists only of intervals of the form $[a,b]$ with $a > b$ (inverted intervals).
\end{Lemma}
 
\begin{proof}
Suppose that $\cap_{i=1}^{d+1} AABB(M(i))$ contains an interval. Without loss of generality, let's assume it is the interval
\[ [\inf\pi_1(\partial D(k)),  \sup\pi_1(\partial D(l))] \]
for some  $k,l\in \{1,2,...,d+1\}$  where $D(k)=\cap_{i\neq k} D_i$. By definition of $\cap_{i=1}^{d+1} AABB(M(i))$, we have $\inf\pi_1(\partial D(k))\geq \inf\pi_1(\partial D(i))$ for all $i\neq k$ and $\sup\pi_1(\partial D(l))\leq \sup\pi_1(\partial D(j))$ for all $j\neq l$. 
Now, consider a hyperplane $P: x_1=p$ with $\inf\pi_1(\partial D(k)) \leq p \leq \sup\pi_1(\partial D(k))$. The hyperplane $P$ intersects $D(j)$ for every $j=1,...,d+1$ (since $P$ cuts through all the boxes $AABB(M(j))$). We also know that $P\cap D_i$ is a $d-1$-dimensional disk for each $i$.  

Therefore, we have a collection of $d$ disks in a $d-1$ dimensional space, and by Helly's Theorem, this collection must have a non-empty intersection. This implies that $D\neq \emptyset$, which contradicts our assumption that $D$ is empty. Hence, all intervals in $\cap_{i=1}^{d+1} AABB(M(i))$ must be inverted intervals of the form $[a,b]$ with $a > b$.
\end{proof}

As an example, we provide the lower bounds of the AABB for a system of three disks in $\mathbb{R}^d$, this is, we compute $\inf\pi_i(\partial D)$ for each $i=1,...,d$ (analogous $\sup\pi_i \partial D$) with $D=\cap_{j=1}^3 D_j$.
Let $M=\{D_1,D_2,D_3\}$ be a Rips system in $\mathbb{R}^d$, then, for each $i=1,...,d$, the computation of the AABB for the disk system $M$ is given by:
\begin{align*}
\inf \pi_i (\partial D) =
\begin{cases}
 \pi_i(c_k-r_k e_i) & \text{if $c_k-r_k e_i \in D$} \\
\max_{j<k}\{\inf \pi_i(\partial D_j \cap \partial D_k)\} & \text{if (*)} \\
\inf \pi_i (\cap_{j=1}^{3} \partial D_j) & \text{otherwise}
\end{cases}
\end{align*}
where (*) denotes the case where  $q\in D$ with $q\in \partial D_j \cap \partial D_k$ such that 
    $$\pi_i(q)=\max_{j<k}\{\inf \pi_i(\partial D_j \cap \partial D_k)\}.$$

The preceding calculations are a result of Lemma \ref{Lemma:extreme values}. It is known that the extremes of the AABB lie in the projections of specific poles, either from a \(d\)-sphere, \((d-1)\)-sphere, or the \((d-2)\)-sphere.





Given a disk system $M$ in $\mathbb{R}^d$, we can compute the minimal axis-aligned bounding box (AABB) for the intersection of all disks in $M$. If the AABB is a point, then it represents the intersection of the disks. This property allows us to identify when the AABB is a point. If the AABB of $M$ is not a point, we can rescale $M$ by a scale factor $\lambda$ such that the AABB of $M_{\lambda}$ becomes a point. The value of $\lambda$ is referred to as the \v{C}ech scale of the system $M$.

\subsection{Minimal axis-aligned bounding box Algorithm.} \label{Section:algorithm.AABB}

Now, we present an algorithm to calculate the minimal axis-aligned bounding box (AABB) of a system of $m$ disks in $\mathbb{R}^d$.
As previously explained, the AABB's extremes are defined by projecting the poles of specific \(i\)-spheres. Thus, in computing the AABB, we will determine the poles for each \(i\)-sphere within the disk system.

\begin{algorithm}
\caption{\texttt{AABB.minimal}}
    \SetKwInOut{Input}{Input}
    \SetKwInOut{Output}{Output}

    \Input{A $d$-disk system $M=\{D_j\}_{j=1}^m$}
    \Output{The minimal AABB for the system $M$}
    Initialize: $P = \emptyset$\\
    \For{$k \leftarrow 1${\normalfont\textbf{ to }}$m$}{
        Let $\mathcal{S}$ be the set of $(d-k+1)$-spheres of $\partial M$\\
        \For{$S${\normalfont\textbf{ in }}$\mathcal{S}$}{
            \For{$q \leftarrow 1${\normalfont\textbf{ to }}$d$}{
                Compute the set $\{s_q^\pm\}$ of $e_q$-poles of $S$\\
                \For{$s \leftarrow \{s_q^\pm\}$}{
                    \If{$s \in \cap_{j=1}^{d} D_j$}{
                        Add: $P \leftarrow s$ \\
                    }
                }
            }
        }
     }
     \eIf{$P \neq \emptyset$}{
        \For{$q \leftarrow 1${\normalfont\textbf{ to }}$d$}{
            $a_q = min_P\{\pi_q(s_q^-)\}$ \\
            $b_q = max_P\{\pi_q(s_q^+)\}$ \\
        }
        \Return($\Pi_{q = 1}^d [a_q, b_q]$)
    }{
        \Return(The disk system does not intersect)
    }
\end{algorithm}

The algorithm starts by initializing the set of poles, \(P\), as an empty set. In each iteration of the first loop, we determine the spheres formed by the intersection of the boundaries of the disk subcollections in the system \(M\) (for this, we require the center and radius, which are computed in Subsection \ref{Subsection:more.sphere.system}). Next, we identify the poles of each sphere and if any of them are present in all the disks of \(M\), we add them to the set \(P\). If \(P\) is not empty, we proceed to calculate the extremes for each dimension of the AABB using the set of north poles for the upper bounds and the set of south poles for the lower bounds. In this case, the output is the product of the intervals defined by the computed extremes. If \(P\) is empty, it indicates that there is no intersection in the disk system.

\section{Acknowledgements}
C.G.E.P acknowledges CONACYT for the financial support provided through a National Fellowship (CVU-638165).

\bibliographystyle{siam}

\begin{thebibliography}{1}

\bibitem{Bell_2019}
{\sc G.~Bell, A.~Lawson, J.~Martin, J.~Rudzinski, and C.~Smyth}, {\em Weighted
  persistent homology}, Involve, a Journal of Mathematics, 12 (2019),
  pp.~823--837.

\bibitem{collision}
{\sc P.~Cai, Y.~Cai, I.~Chandrasekaran, and J.~Zheng}, {\em Collision detection
  using axis aligned bounding boxes}, Simulation, Serious Games and Their
  Applications,  (2013), pp.~1--14.

\bibitem{Espinoza2019}
{\sc J.~F. Espinoza, R.~Hern{\'{a}}ndez-Amador, H.~A.
  Hern{\'{a}}ndez-Hern{\'{a}}ndez, and B.~Ramonetti-Valencia}, {\em A numerical
  approach for the filtered generalized {\v{c}}ech complex}, Algorithms, 13
  (2019), p.~11.

\bibitem{le2015construction}
{\sc N.~K. Le, P.~Martins, L.~Decreusefond, and A.~Vergne}, {\em Construction
  of the generalized cech complex}, 2015.

\bibitem{Mahovsky2004}
{\sc J.~Mahovsky and B.~Wyvill}, {\em Fast ray-axis aligned bounding box
  overlap tests with plucker coordinates}, Journal of Graphics Tools, 9 (2004),
  pp.~35--46.

\bibitem{MAIOLI:2017}
{\sc D.~S. Maioli, C.~Lavor, and D.~S. Gon\c{c}alves}, {\em A note on computing
  the intersection of spheres in $\mathbb{R}^{n}$}, The ANZIAM Journal, 59
  (2017), p.~271–279.

\end{thebibliography}

\end{document}